\documentclass[12pt,letterpaper]{article}

\def\final{1}
\def\llncs{0}

\def\colorson{0}

\usepackage{times}
\usepackage{color}
\usepackage{amsfonts,amsmath}
\usepackage{url}
\ifnum\final=0 \usepackage{showlabels} \fi

 \ifnum\llncs=0

\usepackage{amsthm}
\usepackage{latexsym,amssymb}
\usepackage{graphicx}
\usepackage{wrapfig}
\usepackage{algorithm,algorithmic}

\usepackage[colorlinks=true,breaklinks=true,linkcolor=black,
citecolor=black,filecolor=black,menucolor=black,
pagecolor=black,urlcolor=black]{hyperref}

\advance \topmargin by -\headheight 

\advance \topmargin by -\headsep 

\ifnum\final=0
\else
\fi

\fi 

\ifnum\final=1
\newcommand{\mnote}[1]{}
\else
\newcommand{\mnote}[1]{\marginpar{\tiny \sf {#1}}}
\fi 

\ifnum\colorson=1
\newenvironment{todo}{\noindent
\sf \footnotesize \textcolor{blue}{To go here}:
\begin{CompactItemize}\color{blue}}
{\color{black}\end{CompactItemize}\rm \normalsize}
\else

\fi

\ifnum\llncs=0

\newtheorem{theorem}{Theorem}[section]
\newtheorem{lemma}[theorem]{Lemma}

\newtheorem{definition}{Definition}[section]

\theoremstyle{definition}

\theoremstyle{remark}

\newtheorem{myremark}{Remark} [section]
\newenvironment{remark}{\begin{myremark}}{$\diamondsuit$\end{myremark}}

\newtheorem{myexample}{Example}

\else
\spnewtheorem{protocol}{Protocol}{\bfseries}{\rmfamily}
\spnewtheorem{algm}{Algorithm}{\bfseries}{\rmfamily}
\spnewtheorem{fact}{Fact}{\bfseries}{\rmfamily}
\spnewtheorem{myclaim}{Claim}{\bfseries}{\itshape}
\fi

\newcommand{\lemref}[1]{Lemma~\ref{lem:#1}}

\newcommand{\figref}[1]{Figure~\ref{fig:#1}}

\newcommand{\comment}[1]{}
\newcommand{\ignore}[1]{}

\newenvironment{CompactItemize}{
  \vspace{-4pt}
 \begin{list}{$\bullet$}{%
      \setlength{\leftmargin}{12pt}%
      \setlength{\itemsep}{-2pt}
      }}
  {\end{list}}

\newcommand{\paren}[1]{\left( {#1} \right)}

\newcommand{\re}{\mathbb{R}}

\newcommand{\eps}{\epsilon}

\newcommand{\expe}[2]{\operatorname{E}_{{#1}} \left( {#2} \right)}
\newcommand{\var}[2]{\operatorname{Var}_{{#1}} \left( {#2} \right)}

\newcommand{\bias}{\mathsf{bias}}

\newcommand{\defeq}{\stackrel{{\mbox{\tiny def}}}{=}}

\newcommand{\beq}{\begin{equation}}
\newcommand{\eeq}{\end{equation}}
\newcommand{\bml}{{\begin{multline}}}
\newcommand{\eml}{{\end{multline}}}

\definecolor{brown}{rgb}{0.4,0.2,0}
\definecolor{pink}{rgb}{.9,0,0.3}
\definecolor{darkgreen}{rgb}{0,.3,0}
\definecolor{lightgray}{gray}{.8}
\definecolor{lightrose}{rgb}{1, .7, .7}
\definecolor{lightskyblue}{rgb}{.54,.812,1}
\definecolor{aquamarine}{rgb}{.442,1,.812}
\definecolor{lightgreen}{rgb}{.532,.942,.532}
\definecolor{lightpink}{rgb}{1,.5,0.63}
\definecolor{thistle}{rgb}{0.9,0.749,0.9}
\definecolor{lightblue}{rgb}{0.67843,0.847,0.9}

\newcommand{\etal}{{\em et al.}}

\begin{document}

\title{Efficient, Differentially Private Point Estimators}
\author{Adam Smith\thanks{Department of Computer Science and Engineering, The Pennsylvania State University, University Park, PA, USA. Email: {\tt  asmith at cse.psu.edu}. This work was partly supported by National Science Foundation award \#0729171.}}

\maketitle

\begin{abstract}
{\em Differential privacy} is a recent notion of privacy for statistical databases that provides rigorous, meaningful confidentiality guarantees, even in the presence of an attacker with access to arbitrary side information.\mnote{add sentences here...}

We show that for a large class of parametric probability models, one can construct a differentially private  estimator whose distribution converges to that of the maximum likelihood estimator. In particular, it is efficient and asymptotically unbiased.
This result provides (further) compelling evidence that  rigorous notions of privacy in statistical databases can be consistent with statistically valid inference. 
\end{abstract}

\section{Introduction}

Privacy is a fundamental problem in modern data analysis. 
Increasing volumes of personal and sensitive data are collected  by government and other organizations. The potential social benefits of analyzing these databases are significant; at the same time, releasing information from repositories of sensitive data can cause devastating damage to privacy. The challenge is to discover and release global characteristics of these databases, without compromising the privacy of the individuals whose data they contain. 

There is a vast body of work on this problem in statistics and computer science. However, until recently, most schemes proposed in the literature lacked rigorous analysis of privacy and utility. Few works even formulated a precise definition of their schemes' conjectured properties.

In this paper, we explore the potential of {\em differential privacy}, 
a definition of privacy due to Dwork \etal\ \cite{DMNS06} 
that emerged from a line of work in
cryptography \cite{DiNi03,DwNi04,BDMN05}.  This notion of privacy
makes assumptions neither about what kind of attack might be
perpetrated based on the released statistics, nor about what
additional information the attacker might possess. It resolves a
number of problems present in previous attempts at a
definition. In particular, it provides precise guarantees in the presence of arbitrary side information available to the adversary but unknown to the organization that is releasing information. 

Specifically, 
 we show that for well-behaved parametric probability models, one can construct a differentially-private estimator whose distribution converges to that of the MLE. In particular, it is efficient and asymptotically unbiased.
This provides (further) strong evidence that  rigorous notions of database privacy can be consistent with statistically valid inference.

\paragraph{Differential Privacy}

The problem of identifying which information in the database is safe to release
has generated a vast body of work, both in
statistics and computer science.
Until recently, there were two
nearly disjoint fields studying the data privacy problem:
``statistical disclosure limitation'' (also known as ``data confidentiality''), initiated by the statistics community
in 1960s, and 
``privacy-preserving data mining'',
%
%
 active in the database community during the 1980's and rekindled at the turn of the 21st century by researchers in data mining.  The literature in both fields is far too vast to survey here. For some pointers to the broader literature  in statistics, see~\cite{W65,Dalenius1977,Dalenius1982,Diaconis1998,Slavkovic2004,
Fienberg2004a}. For early work in computer science, see the survey in \cite{AW89}. Recent work in data mining was started by \cite{AS00} and led to an explosion of literature. For (partial) references, see~\cite{CliftonKVLZ02,Gehrke06tutorial,Sweeney05}.

However, the schemes proposed in these fields lack rigorous analysis of privacy. 
Typically, the schemes have either no formal privacy guarantees or
ensure security only against a specific suite of attacks. This
leaves them potentially vulnerable to unforeseen attacks, and makes
it difficult to compare different schemes because each of them is
basically solving a different problem.

A recent line of work~\cite{DiNi03,EGS03,DwNi04,BDMN05,DMNS06,DKMMN06,Dwork06,NRS07,DMT07,MT07,BCDKMT07,RHS07,BLR08,DY08,KS08,KLNRSf08}, called {\em private data analysis}, seeks to place data privacy on more firm theoretical foundations and has been  successful at formulating a strong, yet attainable privacy definition.
The intuition behind the definition, which is due to Dwork~\etal~\cite{DMNS06}, is that whether an individual supplies her actual or fake information has almost no effect on the outcome of the analysis. Roughly, a randomized algorithm that takes sensitive data as input and outputs a product for publication is considered privacy-preserving if databases that differ in one entry induce nearby distributions on its outcomes (see below for a precise definition). 

A number of techniques for designing differentially private algorithms are now known. 
These are surveyed by Dwork~\cite{Dwork07,Dwork08} and Nissim~\cite{Nissim08}. 

\paragraph{Our Contribution}

This paper provides a qualitatively different result from previous work, in that it relates the perturbation added for differential privacy to the provably optimal error of point estimators. For a broad class of problems, we show that differential privacy can be provided at an asymptotically vanishing cost to accuracy.

\newcommand{\infrac}[2]{(#1)/(#2)}
Specifically, we show a modification to the maximum likelihood estimator for parametric models which satisfies differential privacy and is {\em asymptotically efficient}, meaning that the averaged squared error of the estimator is $\infrac{1+o(1)}{nI(\theta)}$, where $n$ is the number of samples in the input, $I(\theta)$ denotes the Fisher information of $f$ at $\theta$ and  and $o(1)$ denotes a function that tends to zero as $n$ tends to infinity. Differential privacy is quantified by a parameter $\eps>0$ which measures information leakage; our estimator satisfies differential privacy with  $\lim_{n\to\infty}\eps=0$.

\section{Definitions}

Consider a parameter estimation problem defined by a model $f(x;\theta)$ where $\theta$ is a real-valued vector in a bounded space $\mathbf{\Theta} \subseteq \re^{p}$ of diameter $\Lambda$, and $x$ takes values in a $D$ (typically, either a real vector space or a finite, discrete set). 

We will generally use the following notational convention: capital latin letters ($X$, $T$, etc) refer to {\em random} variables or processes. Their lower case analogues refer to fixed, deterministic values of these random objects (i.e. scalars, vectors, or functions).

Given i.i.d. random variables $X=(X_{1},...,X_{n})$ drawn according to the distribution $f(\cdot;\theta)$, we would like to estimate $\theta$ using an estimator  $t$ that takes as input the data $x$ as well an additional, independent source of randomness $R$ (used, in our case, for perturbation):

$$\begin{array}{ccccl}\theta & \to & X & \to & t(X,R) = T(X) \\ &  &  &  & \ \uparrow \\ &  &  &  & R 
\end{array}$$

Even for a {\em fixed} input $x=(x_{1},...,x_{n})\in D^{n}$, the estimator $T(x)= t(x,R)$ is a random variable distributed in the parameter space $\re^{p}$. For example, it might consist of a deterministic function value that is perturbed using additive random noise, or it might consist of a sample from a posterior distribution constructed based on $x$. We will use the capital letter $X$ to denote the random variable, and lower case $x$ to denote a specific value in $D^{n}$. Thus, the random variable $T(X)$ is generated from two sources of randomness: the samples $X$ and the random bits used by $T$.

\paragraph{Differential Privacy}

We say two {\em fixed} data sets $x$ and $x'$ in $D^{n}$ are \emph{neighbors} if $x$ and $x'$ agree in all but one position, that is for some $i$, 
$${x=(x_{1},...,x_{i-1}, x_{i}, x_{i+1},...,x_{n}) \atop x'=(x_{1},...,x_{i-1},x_{i}', x_{i+1},...,x_{n})}$$

Differential  privacy compares the distributions of $T(x)= t(x,R)$ and $T(x')= t(x',R)$ corresponding to neighboring data sets $x,x'$. It requires that for all possible pairs of neighboring data sets, the corresponding distributions be close:

\begin{definition}\label{def:diffep}
A randomized algorithm $T(\cdot)$ is \emph{$\eps$-differentially private} if for all neighboring pairs of databases $x$ and $x'$, and for all measurable subsets of outputs (events) $S$:
$$ \Pr(T(x) \in S) \leq e^{\eps}\times \Pr(T(x')\in S)\,.$$
\end{definition}

This condition states that on single point in the input set can significantly influence the distribution of the estimator. 
Note that the privacy condition {\em makes no reference to a distribution on $x$}. It is a ``worst-case'' notion of privacy that provides a guarantee even when our modeling of the distribution on $x$ is incorrect. 

Given two probability measures $p$ and $q$ on a space $\Omega$, we can define the multiplicative distance between $p$ and $q$ to be $$d_\times(p,q) \defeq \ln\left(\sup_{S\in \Omega} \max(\frac{p(S)}{q(S)}, \frac{q(S)}{p(S)} ) \right)= \sup_{S\in \Omega}\left( \ln\left| \frac{p(S)}{q(S)} \right|\right)\, .$$
(We say $d_\times(p,q)=\infty$ when the supremum above doesn't exist.)
Thus, $\eps$-differential privacy requires that, for all fixed  neighboring data sets $x$ and $x'$, the multiplicative distance between (the distributions of) $T(x)$ and $T(x')$ be at most $\eps$. 
The exact choice of distance function significantly affects the practical meaning of differential privacy---see Section 4, Remark 2 in \cite{DMNS06} and \cite{KS08} for discussion. 

\paragraph{The MLE and Efficiency} Many methods exist to measure the quality of a point estimator. In this paper, we consider the expected squared deviation from the real parameter $\theta$. For a one-dimensional parameter ($p=1$), this can be written:
$$J_{T}(\theta) \defeq \expe{\theta}{(T(X)-\theta)^{2}} $$

The notation $\expe{\theta}{...}$ refers to the fact that $X$ is drawn i.i.d. according to $f(\cdot;\theta)$. 
If $T(X)$ is unbiased, then $J_{T}(\theta)$ is simply the variance $\var{\theta}{T(X)}$. Note that all these notations are equally well-defined for a randomized estimator $T(x)=t(x,R)$. The expectation is then also taken over the choice of $R$, i.e. $ J_{T}(\theta) = \expe{\theta}{(t(X,R)-\theta)^{2}}\,.$


(Mean squared error can be defined analogously for higher-dimensional parameter vectors. 
For simplicity we focus here on the one-dimensional case. The development of a higher-dimensional analogue is identical, as long as $p$ is constant with respect to $n$. \mnote{state corresponding result later?})

\newcommand{\mle}{{\hat \theta_{\text{\sc mle}}}}
The maximum likelihood estimator $\mle(x)$ returns a value $\hat \theta$ that maximizes the likelihood  function $L(\theta) = \prod_i f(x_i;\theta)$, if such a maximum exists. It is a classic result that, for well-behaved parametric families, the $\mle$ exists with high probability and is asymptotically normal, centered around the true value $\theta$. Moreover, its expected square error is given by the inverse of Fisher information at $\theta$, 
$$I_f(\theta)\defeq \expe{\theta}{\left[\tfrac{\partial}{\partial \theta} \ln(f(X_1;\theta))\right]^2} \,.
$$ 

\newcommand{\inprob}{{\stackrel{\tiny P}{\longrightarrow}}}
\newcommand{\indist}{{\stackrel{\tiny \mathcal{D}}{\longrightarrow}}}
\newcommand{\bc} {\hat\theta_{bc}}
\renewcommand{\bias}{b}
\newcommand{\bmle}{\bias_{\text{\sc mle}}}

\begin{lemma}\label{lem:mle}
Under appropriate regularity conditions, the MLE converges in distribution to a Gaussian centered at $\theta$, that is  
$\sqrt{n}\cdot (\mle-\theta) \ \indist\ N\left(0,\frac{1}{I_{f}(\theta)} \right)\, .$ 
Moreover, 
$J_{\mle}(\theta)=\frac{1+o(1)}{nI_{f}(\theta)}$, where $o(1)$ denotes a function of $n$ that tends to zero as $n$ tends to infinity. 
\end{lemma}


The MLE has optimal among unbiased estimators; estimators that match this bound are called {\em efficient}.

\begin{definition}
An estimator $T$ is \emph{asymptotically efficient} for a model $f(\cdot;\cdot)$ if, for all $\theta$, the expected squared error converges to that of the MLE, that is, for all $\theta\in \mathbf{\Theta}$, \ \ $ J_{T}(\theta) \leq  \dfrac{1+o(1)}{nI_{f}(\theta)}\, .$
\end{definition}

\paragraph{Bias Correction} 

The asymptotic efficiency of the MLE implies that its bias, $\bmle(\theta) \defeq \expe{\theta}{\mle-\theta}$, goes to zero more quickly than  $1/\sqrt{n}$. However, in our main result, we will need an estimator with much lower bias. This can be obtained via a (standard) process known as bias correction. 

Under appropriate regularity assumptions, we can describe the bias of MLE  precisely, namely 
$$\expe{\theta}{\mle - \theta } = \frac{b_1(\theta)}{n} + O\left(\frac1{n^{3/2}}\right)\,,$$
where $b_1(\theta)$ has a uniformly bounded derivative (see, for example,  discussions in Cox and Hinkley \cite{CL74}, Firth \cite{Firth93}, and Li \cite{Li98}). Several methods exist for correcting this bias. The simplest is to subtract off an estimate of the leading term, using $b_1(\mle)$ to estimate $b_1(\theta)$; the result is called the {\em bias-corrected MLE}, 
$$\bc \defeq \mle -b_1(\mle)/n\,.$$

\begin{lemma}\label{lem:bc}
The bias-corrected MLE  $\bc = \mle - b_1(\mle)/n$, converges at the same rate as the MLE but with lower bias, namely, $$\displaystyle\sqrt{n}\cdot (\bc-\theta) \ \indist\ N(0,\tfrac1{I_{f}(\theta)}) \qquad \text{and}\qquad \bias_{bc} \defeq \expe{\theta}{\bc-\theta}=O(n^{-3/2})\,.$$
\end{lemma}

\section{A Private, Efficient Estimator}

We can now state our main result:
\begin{theorem}
Under appropriate regularity conditions, there exists a (randomized) estimator $T$ which is asymptotically efficient and $\eps$- differentially private, where $\lim_{n\to\infty}\eps = 0$.
\end{theorem}

More precisely, the construction takes as input the parameter $\eps$ and produces an estimator $T$ with mean squared error $\frac{1}{nI_f(\theta)}(1 + O(n^{-1/5} \eps^{-6/5} ))$. Thus, as long as $\eps$ goes to 0 more slowly than $n^{-1/6}$, the estimator will be asymptotically efficient.

The idea is to apply the ``sample-and-aggregate'' method of \cite{NRS07}, similar in spirit to the parametric bootstrap.  The procedure is quite general and can be instantiated in several variants. We present a particular version which is sufficient to prove our main theorem.

The estimator $T^*$ takes the data $x$ as well as a parameter $\eps>0$ (which measures information leakage) and a positive integer $k$ (to be determined later). The idea is to break the input into $k$ blocks of $n/k$ points each, compute the (bias-corrected) MLE on each block, and release the average of these estimates plus some small additive perturbation. The procedure is given in Algorithm 1 and illustrated in \figref{samp-agg}.

\begin{algorithm}[h] 
\caption{On input $x=(x_1,...,x_n) \in D^n$, $\eps>0$ and $k\in\mathbb{N}$:}
\begin{algorithmic}[1]
\STATE Arbitrarily divide the input $x$ into $k$ disjoint sets $B_1,...,B_k$ of $t=\frac n k$ points. \\ We call these $k$ sets the {\em blocks} of the input.
\FOR {each block $B_j = \{x_{(j-1)t +1},...,x_{jt}\}$, }
\STATE Apply the bias corrected MLE $\bc $ to obtain an estimate $z_j = \bc(x_{(j-1)t+1},...,x_{jt})$.
\ENDFOR
\STATE Compute the average estimate: $\bar z = \tfrac{1}{k}\sum_{j=1}^k z_j$.

\STATE Draw a random observation $R$ from a double-exponential (Laplace) distribution with standard deviation $\sqrt{2}\cdot \Lambda / (k\eps)$,
that is,  draw $Y\sim {\sf Lap}\paren{\frac\Lambda{k\eps}}$ where  ${\sf Lap}(\lambda)$ is the  distribution on $\re$ with density $h(y) = \tfrac1{2\lambda} e^{y/\lambda}$. (Recall that $\Lambda$ is the diameter of the parameter space $\mathbf\Theta$.)

\STATE Output $T^*=\bar z +R$.
\end{algorithmic}
\end{algorithm}



The resulting estimator has the form 
\begin{equation}\label{eq:avg}
T^{*}(x) \defeq \paren{\frac{1}{k}\sum_{i=1}^{k} \bc\paren{x_{(i-1)t+1},...,x_{it}} } + {\sf Lap}\paren{\frac\Lambda{k\eps}}\,
\end{equation}

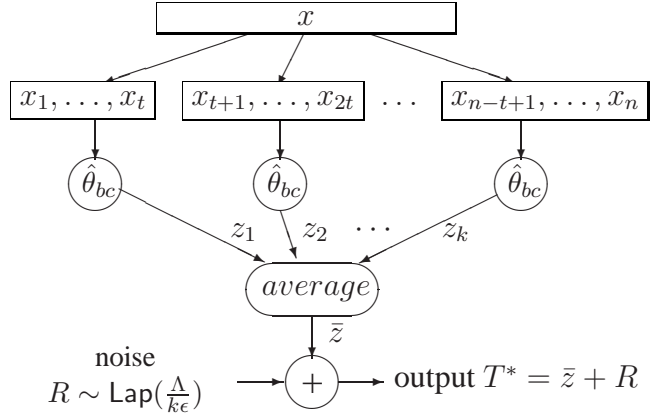
\begin{figure}[tb]
\begin{center}

\newenvironment{overlay}[1]{}{}
\setlength{\unitlength}{1pt}
\begin{picture}(320,160)
\put(156,147){\vector(-3,-1){56}} \put(176,147){\vector(-1,-2){10}}
\put(200,147){\vector(3,-1){56}}

\multiput(97,115)(70,0){2}{\vector(0,-1){16}}
\put(258,115){\vector(0,-1){16}}

\put(106,88){\vector(2,-1){55}} 
\put(167,80){\vector(1,-3){6}}
\put(249,86){\vector(-2,-1){52}}

\put(120,150){\fcolorbox{black}{white}{\hspace*{50pt}$x$\hspace*{50pt}}}

\put(65,120){\fcolorbox{black}{white}{\normalsize$x_{1},\dots,x_{t}$}}
\put(130,120){\fcolorbox{black}{white}{\normalsize$x_{t+1},\dots,x_{2t}$}}
\put(205,120){$\dots$}
\put(228,120){\fcolorbox{black}{white}{\normalsize$x_{n-t+1},\dots,x_{n}$}}

\multiput(97,90)(70,0){2}{\circle{20}} \put(258,90){\circle{20}}
\put(179,50){\oval(50,20)}


\multiput(92,87)(70,0){2}{\small $\bc$}
\put(253,87){\small$\bc$}
 \put(160,48){$average$}

\put(148,70){$z_{1}$}
\put(175,70){$z_{2}$}
\put(195,70){$\bf \cdots$}
\put(228,70){$z_{k}$}

\put(185,30){$\bar z$}


 \put(179,40){\vector(0,-1){16}}
 \multiput(151,15)(38,0){2}{\vector(1,0){17}}

 \put(179,15){\circle{20}}
 \put(175,11){\large +}
\put(70,13){\parbox{2.7cm}{\begin{center}\small {noise} $R\sim{\sf Lap}(\frac{\Lambda}{k\eps})$\end{center}}}

 \put(210,13){\normalsize output $T^*=\bar z +R$}
\color{black}
\end{picture}

\caption{The estimator $T^*$. When the number of bins $k $ is $ o(n^{2/3})$ and $\eps$ is not too small, $T^*$ is asymptotically efficient (\lemref{efficiency}).}
\label{fig:samp-agg}
\end{center}
\end{figure}

\begin{lemma}[\cite{BDMN05,NRS07}]
For any choice of the number of blocks $k$, the estimator $T^{*}$ is $\eps$- differentially private. 
\end{lemma}
The lemma follows from the general techniques developed in \cite{DMNS06,NRS07}; we include a direct proof here for completeness.
\begin{proof}
Fix a  particular value of $x$, and consider the effect of changing a single entry $x_{i}$ to obtain a database $x'$ (for any particular index $i$). At most one of the numbers $z_{j}$ can change, depending on the block which contains $x_{i}$. The number $z_{j}$ that changes can go up or down by at most $\Lambda$, since the parameter takes values in $[0,\Lambda]$. This means that the mean $\bar z$ can change by at most $\Lambda/k$. 

The random variables $T^{*}(x)$ and $T^{*}(x')$ are thus Laplace random variables with identical standard deviations and means differing by at most $\Lambda/k$. Let $h_x$ and $h_{x'}$ be the corresponding density functions. As in \cite{BDMN05,DMNS06}, observe that the ratio  of their densities is at most $e^{\eps}$ since for any real number $y$:
$$\frac{h_x(y)}{h_{x'}(y)} = \frac{\exp(\frac{\eps k}{\Lambda} |y-\bar z|)}
{\exp(\frac{\eps k}{\Lambda} |y-\bar z'|)} \leq \exp(\tfrac{\eps k}{\Lambda}|\bar z - \bar z'|) \leq \exp(\eps) \, .$$
Similarly, the ratio is bounded below by $\exp(-\eps)$. For any measurable set $S\subseteq \re$ with non-zero measure, the ratio $\Pr(T^{*}(x)\in S) / \Pr(T^{*}(x')\in S)$ is thus between $e^{-\eps}$ and $e^{\eps}$. This is exactly the requirement of differential privacy. 
\end{proof}

\begin{lemma}\label{lem:efficiency}
Under the regularity conditions of \lemref{bc}, if  $\eps=\omega(\frac{1}{\sqrt[6] n})$ and $k$ is set appropriately, the estimator $T^{*}$ is  asymptotically unbiased, normal and efficient, that is
$$\sqrt{n} \cdot{T^{*}(X)} \ \indist\ N(\theta,\frac1{I_{f}(\theta)})\quad \text{when }X=X_{1},...,X_{n} \sim f(\cdot,\theta)\text{ are i.i.d.}$$
\end{lemma}

\begin{proof}
We will select $k$ as a function of $n$ and $\eps$. For now, assume that $t=\frac{n}k$ goes to infinity with $n$. By \lemref{mle}, each $Z_{i} =\mle\paren{X_{(i-1)t+1},...,X_{it}}$ converges to normal, and moreover the bias and variance of $Z_i$ can be bounded:
$$\expe{\theta}{Z_{i}} = \theta \pm O\Big(\frac k n\Big)^{3/2}\ \text{and}\ \var{\theta}{\sqrt{\frac n k}\cdot Z_{i}} = \frac{1+o(1)}{I_{f}(\theta)}$$

Consider the averaged estimator $\bar Z = \tfrac 1 k \sum_{i} Z_{i}$. Its expectation is  equal to the expectations of the $Z_{i}$, while its variance scales with $k$:
$$\expe{\theta}{\bar Z} = \expe\theta{Z_{1}} = \theta \pm O\Big(\frac k n\Big)^{3/2}$$
$$\textstyle \var\theta{\bar Z} = \dfrac{1}{k}\var\theta{Z_{i}} = \dfrac{1}{k} \cdot\dfrac{k}{n} \cdot\dfrac{1+o(1)}{I_{f}(\theta)} = \dfrac{1+o(1)}{nI_{f}(\theta)}$$

Recall that the mean squared error $J_{}(\theta)$ is the sum of the variance and squared bias of an estimator. Since the squared bias is $O(k^{3}/n^{3})$, it vanishes asymptotically 
compared to the variance as long as $k = o(n^{{2/3}})$, that is, as long as $k/n^{2/3}\to 0$.

Thus, for sufficiently small $k$, the estimator $\bar Z$ is efficient. We now consider for which values of $k$ the added noise is small enough so that it does not affect the efficiency of $T^*$. The noise added to $\bar Z$ to get $T^{*}$ does not contribute to the bias of the estimator, but does add to the variance. 
Specifically: $\expe{\theta}{T^{*}(X)} = \expe{\theta}{\bar Z} = \theta \pm O\left(\frac k n\right)^{3/2}$ and 
$$\textstyle \var\theta{T^{*}(X)} = \var\theta{\bar Z}  + \dfrac{\Lambda^{2}}{k^{2}\eps^{2}}= \dfrac{1}{n}\left(\dfrac{1+o(1)}{I_{f}(\theta)} +  \dfrac{n\Lambda^{2}}{k^{2}\eps^{2}}
\right)$$

$$J_{T^*}(\theta) = \dfrac{1}{n}\left(\dfrac{1+o(1)}{I_{f}(\theta)} +  \dfrac{n\Lambda^{2}}{k^{2}\eps^{2}} + \frac{k^3}{n^2}
\right)$$

If $\eps = \omega(n^{{-1/6}})$  then we can choose $k$ to ensure that $nJ_{T^{*}}(\theta) \to 1/I_{f}(\theta)$. We need $k = o(n^{2/3})$ to get sufficiently small bias and $k = \omega(\frac{\sqrt{n}}{\eps})$ to get the variance of the noise sufficiently low. Taking $k  = \lceil\frac{n^{3/5} \Lambda^{2/5} }{\eps^{2/5}}\rceil$ yields an asymptotic relative error that tends to 1, namely: 

$$ J_{T^*}(\theta) = \dfrac{1}{n}\left(\dfrac{1+o(1)}{I_{f}(\theta)} + O\left(\frac{\Lambda^{6/5}}{n^{1/5}\eps^{6/5}}\right)
\right)
$$ Since $\Lambda$ is constant with respect to $n$, $T^*$ is efficient as long as $\eps n^{1/6}\to \infty$, as desired. \mnote{add proof that $\bar Z$ converges to Normal.}
\end{proof}
%

%



\subsection*{Acknowledgements}

I am grateful to many colleagues in both statistics and computer science for helpful discussions about this project. I would especially like to thank Bing Li, from Penn State's  Department of Statistics, for insightful conversations about bias correction and asymptotic expansions of statistical functionals.

\small
\bibliographystyle{abbrv}
\bibliography{../bib/privacy,../bib/nsf05542Slavkovic,../bib/pacparity}

\begin{thebibliography}{10}

\bibitem{AW89}
N.~R. Adam and J.~C. Wortmann.
\newblock Security-control methods for statistical databases: a comparative
  study.
\newblock {\em ACM Computing Surveys}, 25(4), 1989.

\bibitem{AS00}
R.~Agrawal and R.~Srikant.
\newblock Privacy-preserving data mining.
\newblock In W.~Chen, J.~F. Naughton, and P.~A. Bernstein, editors, {\em SIGMOD
  Conference}, pages 439--450. ACM, 2000.

\bibitem{BCDKMT07}
B.~Barak, K.~Chaudhuri, C.~Dwork, S.~Kale, F.~McSherry, and K.~Talwar.
\newblock Privacy, accuracy, and consistency too: a holistic solution to
  contingency table release.
\newblock In L.~Libkin, editor, {\em PODS}, pages 273--282. ACM, 2007.

\bibitem{BDMN05}
A.~Blum, C.~Dwork, F.~McSherry, and K.~Nissim.
\newblock Practical privacy: The {SuLQ} framework.
\newblock In {\em PODS}, 2005.

\bibitem{BLR08}
A.~Blum, K.~Ligett, and A.~Roth.
\newblock A learning theory approach to non-interactive database privacy.
\newblock In {\em Symposium on the Theory of Computing (STOC)}, 2008.

\bibitem{CliftonKVLZ02}
C.~Clifton, M.~Kantarcioglu, J.~Vaidya, X.~Lin, and M.~Y. Zhu.
\newblock Tools for privacy preserving data mining.
\newblock {\em SIGKDD Explorations}, 4(2):28--34, 2002.

\bibitem{CL74}
D.~R. Cox and D.~V. Hinkley.
\newblock {\em Theoretical Statistics}.
\newblock Chapman-Hall, 1974.

\bibitem{Dalenius1977}
T.~Dalenius.
\newblock Towards a methodology for statistical disclosure control.
\newblock {\em Statistik Tidskrift}, (5):35--64, 1977.

\bibitem{Dalenius1982}
T.~Dalenius and S.~Reiss.
\newblock Data-swapping: A technique for disclosure control.
\newblock {\em Journal of Statistical Planning and Inference}, (6):73--85,
  1982.

\bibitem{Diaconis1998}
P.~Diaconis and B.~Sturmfels.
\newblock Algebraic algorithms for sampling from conditional distributions.
\newblock {\em The Annals of Statistics}, 26(1):363--397, 1998.

\bibitem{DiNi03}
I.~Dinur and K.~Nissim.
\newblock Revealing information while preserving privacy.
\newblock In {\em PODS}, pages 202--210, 2003.

\bibitem{Dwork06}
C.~Dwork.
\newblock Differential privacy.
\newblock In {\em ICALP}, LNCS, pages 1--12, 2006.

\bibitem{Dwork07}
C.~Dwork.
\newblock An ad omnia approach to defining and achieving private data analysis.
\newblock In F.~Bonchi, E.~Ferrari, B.~Malin, and Y.~Saygin, editors, {\em
  PinKDD}, volume 4890 of {\em Lecture Notes in Computer Science}, pages 1--13.
  Springer, 2007.

\bibitem{Dwork08}
C.~Dwork.
\newblock Differential privacy: A survey of results.
\newblock In M.~Agrawal, D.-Z. Du, Z.~Duan, and A.~Li, editors, {\em TAMC},
  volume 4978 of {\em Lecture Notes in Computer Science}, pages 1--19.
  Springer, 2008.

\bibitem{DKMMN06}
C.~Dwork, K.~Kenthapadi, F.~McSherry, I.~Mironov, and M.~Naor.
\newblock Our data, ourselves: Privacy via distributed noise generation.
\newblock In {\em EUROCRYPT}, pages 486--503, 2006.

\bibitem{DMNS06}
C.~Dwork, F.~McSherry, K.~Nissim, and A.~Smith.
\newblock Calibrating noise to sensitivity in private data analysis.
\newblock In S.~Halevi and T.~Rabin, editors, {\em TCC}, volume 3876 of {\em
  Lecture Notes in Computer Science}, pages 265--284. Springer, 2006.

\bibitem{DMT07}
C.~Dwork, F.~McSherry, and K.~Talwar.
\newblock The price of privacy and the limits of {LP} decoding.
\newblock In D.~S. Johnson and U.~Feige, editors, {\em STOC}, pages 85--94.
  ACM, 2007.

\bibitem{DwNi04}
C.~Dwork and K.~Nissim.
\newblock Privacy-preserving datamining on vertically partitioned databases.
\newblock In {\em CRYPTO}, pages 528--544, 2004.

\bibitem{DY08}
C.~Dwork and S.~Yekahnin.
\newblock On lower bounds for noise in private analysis of statistical
  databases.
\newblock Presentation at BSF/DIMACS/DyDan Workshop on Data Privacy, February
  2008.

\bibitem{EGS03}
A.~V. Evfimievski, J.~Gehrke, and R.~Srikant.
\newblock Limiting privacy breaches in privacy preserving data mining.
\newblock In {\em PODS}, pages 211--222, 2003.

\bibitem{Fienberg2004a}
S.~E. Fienberg and A.~B. Slavkovic.
\newblock Making the release of confidential data from multi-way tables count.
\newblock {\em Chance}, 17(3), 2004.

\bibitem{Firth93}
D.~Firth.
\newblock Bias reduction of maximum likelihood estimates.
\newblock {\em Biometrika}, 80(1):27--38, 1993.

\bibitem{Gehrke06tutorial}
J.~Gehrke.
\newblock Models and methods for privacy-preserving data publishing and
  analysis (tutorial slides).
\newblock In {\em Twelfth Annual SIGKDD International Conference on Knowledge
  Discovery and Data Mining (SIGKDD 2006)}, 2006.

\bibitem{KLNRSf08}
S.~P. Kasiviswanathan, H.~K. Lee, K.~Nissim, S.~Raskhodnikova, and A.~Smith.
\newblock What can we learn privately?
\newblock In {\em FOCS}, 2008, To Appear.

\bibitem{KS08}
S.~P. Kasiviswanathan and A.~Smith.
\newblock A note on differential privacy: Defining resistance to arbitrary side
  information.
\newblock {\em CoRR}, arXiv:0803.39461 [cs.CR], 2008.

\bibitem{Li98}
B.~Li.
\newblock An optimal estimating equation based on the first three cumulants.
\newblock {\em Biometrika}, 85(1):103--114, 1998.

\bibitem{MT07}
F.~McSherry and K.~Talwar.
\newblock Differential privacy in mechanism design.
\newblock In A.~Sinclair, editor, {\em IEEE Symposium on the Foundations of
  Computer Science (FOCS)}, October 2007.

\bibitem{Nissim08}
K.~Nissim.
\newblock Private data analysis via output perturbation.
\newblock In C.~C. Aggarwal and P.~S. Yu, editors, {\em Privacy-Preserving Data
  Mining: Models and Algorithms}, pages 383--413, 2008.

\bibitem{NRS07}
K.~Nissim, S.~Raskhodnikova, and A.~Smith.
\newblock Smooth sensitivity and sampling in private data analysis.
\newblock In U.~Feige, editor, {\em Symposium on the Theory of Computing
  (STOC)}, 2007.

\bibitem{RHS07}
V.~Rastogi, S.~Hong, and D.~Suciu.
\newblock The boundary between privacy and utility in data publishing.
\newblock In {\em VLDB}, pages 531--542, 2007.

\bibitem{Slavkovic2004}
A.~Slavkovic.
\newblock {\em Statistical Disclosure Limitation Beyond the Margins:
  Characterization of Joint Distributions for Contingency Tables.}
\newblock Ph.D. Thesis, Department of Statistics, Carnegie Mellon University,
  2004.

\bibitem{Sweeney05}
L.~Sweeney.
\newblock Privacy-enhanced linking.
\newblock {\em SIGKDD Explorations}, 7(2):72--75, 2005.

\bibitem{W65}
S.~L. Warner.
\newblock Randomized response: A survey technique for eliminating evasive
  answer bias.
\newblock {\em Journal of the American Statistical Association},
  60(309):63--69, 1965.

\end{thebibliography}

\end{document}